\documentclass[aps,prl,reprint,superscriptaddress]{revtex4-1}

	\usepackage{graphicx}% Include figure files
	\usepackage{dcolumn}% Align table columns on decimal point
	\usepackage{hyperref}
	\usepackage{physics}
	\usepackage{amsthm} % Theorem environment
	\usepackage{algorithm}
	\usepackage[noend]{algpseudocode}
	\usepackage{bm,braket}% bold math
	\usepackage{amssymb}% Additional symbols
	\usepackage{soul,xcolor}

	\usepackage{color}
	\hypersetup{colorlinks=true, linkcolor=blue, citecolor=red, urlcolor=blue}

	\theoremstyle{plain}
	\newtheorem{thm}{Theorem}

	\theoremstyle{definition}

	\newtheorem{exmp}{Example}

	\theoremstyle{remark}
	\newtheorem*{rem}{Remark}

\begin{document}

\title{Quantum Speedup in Adaptive Boosting of Binary Classification}

\author{Ximing Wang}
\affiliation{Shenzhen Institute for Quantum Science and Engineering and 
Department of Physics, University of Science and Technology, Shenzhen 518055,
China}
\affiliation{University of Technology Sydney, Australia}
\author{Yuechi Ma}
\affiliation{Shenzhen Institute for Quantum Science and Engineering and 
Department of Physics, University of Science and Technology, Shenzhen 518055,
China}
\affiliation{Center for Quantum Information, Institute for Interdisciplinary
Information Sciences, Tsinghua University, Beijing 100084, China}
\author{Min-Hsiu Hsieh}
\email{min-hsiu.hsieh@uts.edu.au}
\affiliation{University of Technology Sydney, Australia}
\author{Manhong Yung}
\email{yung@sustc.edu.cn}
\affiliation{Shenzhen Institute for Quantum Science and Engineering and 
Department of Physics, University of Science and Technology, Shenzhen 518055,
China}
\affiliation{Shenzhen Key Laboratory of Quantum Science and Engineering,
University of Science and Technology, Shenzhen 518055, China}
\affiliation{Central Research Institute, Huawei Technologies, Shenzhen 51829, China}

\begin{abstract}
In classical machine learning, a set of weak classifiers can be adaptively combined to form a strong classifier for improving the overall performance, a technique called adaptive boosting (or AdaBoost). However, constructing the strong classifier for a large data set is typically resource consuming. Here we propose a quantum extension of AdaBoost, demonstrating a quantum algorithm that can output the optimal strong classifier with a quadratic speedup in the number of queries of the weak classifiers. Our results also include a generalization of the standard AdaBoost to the cases where the output of each classifier may be probabilistic even for the same input. We prove that the update rules and the query complexity of the non-deterministic classifiers are the same as those of deterministic classifiers, which may be of independent interest to the classical machine-learning community. Furthermore, the AdaBoost algorithm can also be applied to data encoded in the form of quantum states; we show how the training set can be simplified by using the tools of t-design. Our approach describes a model of quantum machine learning where quantum speedup is achieved in finding the optimal classifier, which can then be applied for classical machine-learning applications.
\end{abstract}

	%\keywords{AdaBoost, Machine Learning, Quantum Computation}%Use showkeys
	%class option if keyword
	\maketitle

%\section{Introduction}
\textit{Introduction---}
One of the most fundamental topics in machine learning is how to train a machine to perform classification tasks given a set of labelled samples.  Unlike human brains that can efficiently distinguish different objects with eyes, the classification capability of a learning machine is restricted by the amount of available information. In many practical circumstances, the performances of the classifiers could be quite weak, say, they are only slightly better than tossing a fair coin. A big question that attracted significant interests in the past is whether these weak classifiers can be efficiently combined to a strong classifier. Yoav Freund and Robert Schapire~\cite{Freund1996} provided a positive solution to this question in their G\"{o}del Price winning work. Their pioneering work triggered a long line of follow-up work, and led to a well-known adaptive boosting algorithm that yields a strong classifier from an ensemble of weak classifiers, abbreviated as AdaBoost~\cite{Mohri2012}.
%Supervised learning is a branch of machine learning which learns from some prior known labeled samples in order to label new inputs. It is widely used for the task of classification. 
%For instance, AdaBoost~\cite{Mohri2012}, an	abbreviation for ``Adaptive Boosting", is a well-known boosting algorithm for supervised learning that obtains a strong classifier from an ensemble of weak classifiers. It was originally proposed by Yoav Freund and Robert Schapire~\cite{Freund1996} in their G\"{o}del Price winning work. 
The original AdaBoost works well for binary classification problems, such as detecting human in images and videos~\cite{kong2015, Markoski2015}. It can also be extended to multi-class classification tasks, e.g. face recognition~\cite{Owusu2014, Jiang2015}. 

Adaboost also works well with a large amount of widely used machine learning algorithms, such as support vector machine (SVM)~\cite{Li2008},	decision tree~\cite{Roe2005}, rotation forest~\cite{Zhang2008} and so on.	With the help of these algorithms, it has been applied to a variety of scenarios since its appearance, including chemical and biological informatics, financial analysis, information extraction, computer vision and natural language processing, which are all important techniques that are spreading influences into modern industry. However, as most machine learning algorithms, AdaBoost requires a huge amount of data to train, and the complexity of the algorithm grows rapidly with the increase of sample size. This rapidly growing data unambiguously becomes a fundamental challenge for further development. Therefore, any improvement of	the complexity of machine learning algorithms would be extremely valuable.

As a novel field lying at the intersection of quantum physics and computer science, quantum computing shows a potential of significantly improving classical algorithms. Quantum computation, taking advantages of the laws of quantum physics, has the ability to process data with much fewer registers than classical computing because different states of registers could exist at the same time as a superposition. Efficient quantum algorithms have been discovered since 1990s, such as Grover searching	algorithm~\cite{Grover1996}, quantum phase estimation	algorithm~\cite{Kitaev1995}, factoring integers, discrete	logarithm~\cite{Shor1997}, and solving linear system	~\cite{Harrow2009,Childs2017}. Vast new algorithms are constructed based on	them~\cite{Coles2018}, and many of them are proved to be faster than any possible classical, and even probabilistic algorithms~\cite{Buhrman2002}.

Considering the advantage of quantum computation, it is natural to	investigate this new idea into the field of machine learning, seeking	solutions to the ``big data" challenge. Quantum machine learning, as an interdisciplinary field between machine learning and quantum computation, explores how to deal with big data with quantum computers. It has made great progress in recent years~\cite{Biamonte2017, Ciliberto2018}. It has become a matter of interest for the great potential of solving the challenge of ``big	data''~\cite{Cai2015,Li2015,Dunjko2016,Schuld2017}. Major breakthroughs	include quantum support vector	machine~\cite{Rebentrost2014,Li2015}, quantum discriminant	analysis~\cite{Cong2016} and quantum principal component	analysis~\cite{Lloyd2014}.

In this Letter, we revisit the original AdaBoost algorithm, and propose a quantum generalization with a quadratic speedup. Besides a quantum AdaBoost algorithm, we also show that the original AdaBoost algorithm has the similar performance even if the basis classifiers are relaxed to probabilistic ones. Our quantum AdaBoost has the potential to improve various existing quantum machine learning algorithms, such as quantum state discriminating	algorithms~\cite{Sasaki2001,Bae2015}.

%%%%%%%%%%%%%% Review of  Supervised Learning %%%%%%%%%%%%%%%%%%
%\section{Supervised Learning}
%%%%%%%%%%%%%%%%%%%%%%%%%%%%%%%%%%%%%%%%%%%%%%%%%%%%%%%%%%%%%%%%

\textit{Binary Classification---} Classification is the task of assigning same labels to a collection of inputs with the same structures. 
%For example, one may want to identify if a randomly-chosen picture drawn from the database contains a cat or a dog. 
It is a crucial component in constructing supervised learning models. 
%The purpose of supervised learning is to efficiently label any input drawn from an ensemble of data based a training set of data, called classification. For example, one may want to identify if a randomly-chosen picture drawn from the database contains a cat or a dog. The task of classification can be achieved with supervised learning. Ideally, the ensemble is assumed to be fixed, which means that (i) the occurrence	probability of each input is \textit{unchanged}, and (ii) the ensemble taken	for applications should be identical to the ensemble where the training	samples are drawn from.
%Generally, if one randomly picks a sufficiently-large number of images from	the database as a sample, and label them manually, then according to the law of large number in probability theory, the distribution of each feature in the sample should be close to that in the whole database. Therefore, an algorithm works well in such sample should also works well for the whole	database. Furthermore, the problem of classifying the whole database can be reduced to the problem of finding an algorithm that can classify the sample well. In machine learning, this is achieved by optimizing a parameterized learning model.
Specifically, in binary supervised learning, we are given a sample consisting of $N$ labeled examples (data) $\{(x_1,y_1), \cdots, (x_j,y_j), \cdots (x_{N},y_{N})\}$, where each $x_j$ is a data point in a sample space $\Omega$, and $y_j\in \{1,-1\}$ is the label of $x_j$. 
%For example, $x_j$ may be an image randomly picked from the gallery $\Omega$, and $y_j$ is the label such that $y_j = 1$ if $x_j$ contains a cat, 	and $y_j = -1$ otherwise. 
A binary classifier can be considered as a mapping $H:\Omega \to \{-1, +1\}$. 
%The task of supervised learning is to construct such a classifier that can produce as many as possible correct labels  $x_j\to y_j$ in the sample.

However, in many cases, such an optimal classifier is hard to obtain, and only {\it weak} classifiers (or basis classifiers), whose performance may be slightly better then random guessing, are available. In \cite{Kearns1994}, Kearns and Valiant proposed the following  fundamental question: how to {\it boost} a set of weak classifiers	$\{H_1,H_2,\cdots,H_T\}$ to become a strong classifier? It turns out that such a task can be achieved with the boosting technique.

	%%%% Review of Classical AdaBoost%%%%%%%%
%\section{Conventional AdaBoost}
%%%%%%%%%%%%%%%%%%%%%%%%%%%%%%%%%%%%%%%%%%%%%%%%
\textit{Conventional AdaBoost---} The main idea of boosting is to construct
a strong classifier $H_{\textbf{strong}}$ given a collection of weak classifiers $\{H_{1},\cdots, H_T\}$:
	\begin{equation}\label{eq:AdaBoostModel}
		H_{\textbf{strong}}(x)
		= \textbf{sgn} (g_T(x)) \ ,
	\end{equation}
where $g_T(x):=\sum_{t=1}^{T}\alpha_t H_t(x)$ is a weighted sum of the $T$
	basis classifiers $H_t(x)$, and the sign function $\textbf{sgn}(v)$ gives
	$+1$ if $v$ is positive and $-1$ otherwise. 
	
	Assume that each input $x\sim D$ occurs with a probability
	denoted by $p(x)$. The corresponding cost function $C_T$ in boosting is the
	exponential error of $g_T(x)$,
	$ C_T = \sum_x \prod_{t = 1}^T {p\qty(x) e^{-\alpha _t H_t(x)y(x)} }$,
	where the $T$ coefficients $\{\alpha_i\}_{i\in [T]}$ in $g_T(x)$ have to be
	optimized. In the binary case, $H_t(x), y(x) \in \qty{-1, +1}$, we have
	$H_t(x)y(x) = (-1)^{r^x_t}$, where the indicator function $r^x_t$ is $1$ if
	the $t^\text{th}$ classifier $H_t$ classifies $x$ wrong, and $0$ otherwise.
	With the indicator $r^x_t$, $C_T$ can be rewritten as
	\begin{equation}
		\label{eq:ConventionalCost}
		C_T=\sum_x {\prod_{t = 1}^T {
			p(x){e^{-\alpha _t (-1)^{r^x_t}}}
		}} \ .
	\end{equation}

	Note that, an optimal solution $\{\alpha_i\}_{i\in[T]}$ exists because this
	is a convex optimization problem \cite{Mohri2012}. Moreover, this optimal
	solution would not be worse than any of the basis classifiers, as the
	choice of $\alpha_i =1$ and $\alpha_{t\neq i}=0$ reduces to a basis
	classifier, and this would not be better than the optimal solution. As a result, the strong
	classifier $H_{\textbf{strong}}$ should be able to classify (i.e., $g_T(x)
	y(x) > 0$) the inputs with a success rate better than any of the basis
	classifiers.

	The key feature of AdaBoost based on $C_T$ is that one can adaptively
	compute $\alpha_t$ using the recursive relation:
	$C_t=C_{t-1}e^{\alpha_t H_t(x) y(x)}$. We can determine the desired
	coefficient $\alpha_t$ that minimizes $C_t$ by differentiating
	$\frac{\partial C_t}{\partial \alpha_t}$ and bring it to zero.

	As a result, the optimal solution to the conventional AdaBoost is given by
	$\alpha_t=\frac{1}{2} \ln(\frac{1-R_t}{R_t})$ \cite{Mohri2012},
	where the weighted errors
	$R_t:=\underset{x\sim D}{\mathbb{E}}[W^x_t r^x_t]$.
	The weights $W^x_t$ can be obtained adaptively with the information of
	$R_1,\cdots,R_{t-1}$ (details can be found in
	section~\hyperref[sec:TO]{Proof of Thorem~\ref{thm:ProbOS}} in
	appendix or \cite{Mohri2012}). In next section, we shall show that the
	conventional AdaBoost is simply a special case of probabilistic AdaBoost.

%%%%%%%%%%%%%%%%%%%%%%%%%%%%%%%%%%%%%%%%%%%%%%%%
%\section{Prob AdaBoost}
%%%%%%%%%%%%%%%%%%%%%%%%%%%%%%%%%%%%%%%%%%%%%%%%
	\textit{Probabilistic AdaBoost---} Here we present a generalized AdaBoost
	method which can be applied to the cases where each basis classifiers are
	probabilistic. In conventional AdaBoost, each basis classifier $H_t$ must
	produce the same label for each input $x$, even if it is incorrect. We
	consider the setting where each basis classifier produces a label
	\textit{probabilistically} for each input $x$. Specifically, we define
	$q_t(r^x_t| x)$ to be the conditional probability where $H_t$ produces the
	label such that $r^x_t \in \{0,1\}$.

	For further convenience here we define a binary string of length $t$ as
	$\pmb{s}_{t} := s_1 s_2\cdots s_t$ to record the results for a sequence of
	basis classifiers for input $x$. The joint probability to obtain the string
	$\pmb{s}_{t}$ with input $x$ is given by
	$q(\pmb{s}_t,x):= p(x)\prod_{i=1}^t q_i(r^x_i = s_i|x)$.
	With a little abuse of notation, the cost function (exponential error
	\eqref {eq:ConventionalCost}) is redefined for probabilistic case:
	\begin{equation}
		\label{eq:ProbCost}
		C_T = \sum_{x,\pmb{s}_T}
		\prod_{t=1}^T q(\pmb{s}_T,x)\prod_{s_i}e^{-\alpha_t \qty(-1)^{s_t}} \ .
	\end{equation}

	We shall show that (i) the solution of the probabilistic models of AdaBoost
	is a generalization of the conventional case, but (ii) their query
	complexities are the same, which are summarized by the following theorems:
	\begin{thm}[Optimal solution to probabilistic case]
		\label{thm:ProbOS}
		The optimal solution to the probabilistic AdaBoost model, in terms of
		$\{\alpha_t\}$ minimizing the cost function $C_T$ \eqref{eq:ProbCost},
		is given by
		\begin{equation}
			\alpha_t
			=
			\frac{1}{2}
			\ln
			\frac
			{1 - \tilde{R}_t}
			{\tilde{R}_t} \ ,
		\end{equation}
	where $\tilde{R}_t$ is the weighted error taking over the joint distribution
	$q(\pmb{s}_t,x)$:
	\begin{equation}
		\label{eq:PError}
		\tilde{R}_t = {\mathbb{E}}\qty[W^x_{\pmb{s}_t}r^x_t]
		= \sum_{x,\pmb{s}_t} q(\pmb{s}_t,x) W^x_{\pmb{s}_t}r^x_t \ .
	\end{equation}
	\end{thm}
	Here the weights $W^x_{\pmb{s}_t}$ are obtained adaptively: starting from
	empty string $\pmb{s}_0$, $W^x_{\pmb{s}_0}\equiv 1$, for $k\in[0,t-1]$
	\begin{equation}
	\label{eq:Updating}
		\begin{aligned}
			W^x_{\pmb{s}_{k+1}}              & =
			\frac{1}{2 \tilde{R}_k}W^x_{\pmb{s}_k}   & \text{if }\ s_{k+1}=1 \\
			W^x_{\pmb{s}_{k+1}}             & =
			\frac{1}{2(1-\tilde{R}_k)}W^x_{\pmb{s}_k}& \text{if }\ s_{k+1}=0 \ .
		\end{aligned}
	\end{equation}

	The proof of theorem~\ref{thm:ProbOS} is given in the
	appendix. The result of conventional AdaBoost is just a special case of
	theorem~\ref{thm:ProbOS} while $q_i(r^x_i|x)$ only take values in
	$\qty{0,1}$.

	Defining $\hat{c}=\max_{x,\pmb{s}_t}\qty{W^x_{\pmb{s}_t}}$ as the maximum
	value of $W^x_{\pmb{s}_t}$ for all $x$, the complexity to find the optimal
	solution is as follows.

	\begin{thm}[Query Complexity for probabilistic case]
		\label{thm:CompProb}
		There exists an algorithm that can approximate the optimal coefficients
		in theorem~\ref{thm:ProbOS} with precision $\epsilon$ with
		$\mathcal{O}\qty(\frac{\hat{c}^2}{\epsilon^2}T)$ queries of basis
		classifiers $H$, and requires
		$N=\mathcal{O}(\frac{\hat{c}^2} {\epsilon^2})$ data points as training
		sample.
	\end{thm}

	\begin{proof}
		(Details see appendix.) The explicit algorithm is shown in
		algorithm~\ref{alg:Classical}. This is indeed a Monte Carlo method.
		Note that, although there are $2^T$ possible routines of $\pmb{s}_T$
		for each $x$, according to a variation of Hoeffding's inequaltiy
		(theorem~\ref{thm:MDinequal}), the precision of sampling
		$\tilde{R}_t = \mathbb{E}[W^x_{\pmb{s}_t}r^x_t]$ only depends on the
		maximum and minimum value of the function to be averaged, but not the
		distribution. This can be seen as each input $x$ is only evaluated once
		with each classifier $H_t$, for each branch $\pmb{s}_T$, the
		evaluations behaves as if the classifiers are deterministic.
		Furthermore, each routine $\pmb{s}_T$ obeys the same update rule;
		therefore the maximum value is bounded in the same way, i.e., only a
		sample $S$ of size $N = \mathcal{O}\qty(\frac{\hat{c}^2}{\epsilon^2})$
		examples are required to estimate $\tilde{R}_t$ with $\hat{R}_t:=
		\frac{1}{N}\sum_{x\in S}W^x_{\pmb{s}_t} r^x_t$ with precision
		$\epsilon$. To complete the computation, this need to be repeated for
		each $t$, and gives the total query complexity
		$\mathcal{O}\qty(NT)=\mathcal{O}\qty(\frac{\hat{c}^2}{\epsilon^2}T)$.
	\end{proof}

	% The Classical Algorithm
	\begin{algorithm}[H]
		\caption{Classical Adaboost}\label{alg:Classical}
	\begin{algorithmic}[1]
		\State Import $H_t$; \Comment The $T$ basis classifiers
		\State Input $S$; \Comment The sample of size $N$

		% \BlankLine
		% \Init{\\
			\State Initialize $W^x_0 \equiv 1$;
				% \State \Comment Iterate over all classifiers, O(T) times
			\For{$t$ from $1$ to $T$} \label{line:QueryStart}
				% \State \Comment Evaluate the error of a classifier, O(N) times
				\For{$x$ in $S$}
					% \State $R[t] \leftarrow R[t]+r(S[x],t)*W[x]$;
					\State $r^x_t \leftarrow H_t(x)$;
				\EndFor
			\EndFor \label{line:QueryEnd}
			\State \Comment Iterate over all classifiers, only arithmetic part
			below.
			\For{$t$ from $1$ to $T$}
				\State $\hat{R}_t\leftarrow
				\frac{1}{\abs{S}}\sum_{x\in S} r^x_t W^x_{\pmb{s}_t}$
				\Comment Take the average over $x$ \label{line:Average}
				% \State \Comment Update the weights for next classifier
				\For{$x$ in $S$} \label{line:start}
					\If{$r^x_t = 0$}
						\State $W^x_{\pmb{s}_t} \leftarrow{W^x_{\pmb{s}_{t-1}}}
						/{2(1-\hat{R}_t)}$;
					\Else \
						\State $W^x_{\pmb{s}_t} \leftarrow{W^x_{\pmb{s}_{t-1}}}
						/{2\hat{R}_t}$;
					 \EndIf
				\EndFor \label{line:end}
				\State $\alpha_t \leftarrow \frac{1}{2}
				\ln\left(\frac{1-\hat{R}_t}{\hat{R}_t}\right)$;
			\EndFor
			\State Output $\qty{\alpha_1,\cdots,\alpha_{T}}$;
	\end{algorithmic}
	\end{algorithm}

%%%%%%%%%%%%%%%%%%%%%%%%%%%%%%%%%%%%%%%%%%%%%%%%%%%%
%\section{Quantum AdaBoost}
%%%%%%%%%%%%%%%%%%%%%%%%%%%%%%%%%%%%%%%%%%%%%%%%%%%%

	\textit{Quantum AdaBoost---} Here we propose a quantum version of the
	AdaBoost algorithm, which provides a quadratic speed up of the query
	complexity in term of the sample size $N$.

	A key observation about the classical AdaBoost algorithm~\ref{alg:Classical}
	is that, the weights $W^x_{\pmb{s}_t}$ are updated independently for each
	$x$; also we are only interested in the average
	$\tilde{R}_t = \mathbb{E}[W^x_{\pmb{s}_t} r^x_t]$ \eqref{eq:PError}.

	To translate algorithm~\ref{alg:Classical} into a quantum algorithm, first
	we define $\ket{\pmb{s}_t}_{\mathcal{R}_t}\equiv
	\ket{s_1}\otimes\cdots\otimes\ket{s_t}$ to be a register encoding the
	string $\pmb{s}_t$. Provided that one can access each basis classifier as a
	query operator $\hat{\mathcal{H}}_i$ in quantum superposition, i.e.,
	\begin{equation}\label{eq:QQuery}
	\hat{\mathcal{H}}_i\ket{x}_X\ket{0} \to
			\ket{x}_X
			\left(
				\sqrt{q_i(0|x)}\ket{0}
				+
				\sqrt{q_i(1|x)}\ket{1}
			\right) \ ,
	\end{equation}
	algorithm~\ref{alg:Classical} can be sped up by applying phase
	estimation. The $q_i(r^x_i|x)$ here works as the same role as the
	probabilistic case. Also, a register $\ket{W^x_{\pmb{s}_t}}_M$ encoding the 
	numerical value of weights $W^x_{\pmb{s}_t}$ need to be introduced.

	The algorithm~\ref{alg:Classical} can be divided into three parts:

	In lines~\ref{line:QueryStart}-\ref{line:QueryEnd}, it queries each
	classifier with each example in sample $S$. This can be done with $t$
	queries $\bigotimes_{i=1}^t \hat{\mathcal{H}}_i$ on each qubit in
	$\mathcal{R}_t$, one can obtain a state of the superposition of all 
	branches $\pmb{s}_t$:
	\begin{equation}
		\label{eq:allQuery}
		\bigotimes_{i=1}^t \hat{\mathcal{H}}_i\ket{x}_X\ket{0}_{\mathcal{R}_t}
		\to
		\sum_{\pmb{s}_t}\sqrt{q(\pmb{s}_t,x)}
		\ket{x}_X\ket{\pmb{s}_t}_{\mathcal{R}_t}.
	\end{equation}

	Then in line~\ref{line:Average}, the algorithm evaluates $\hat{R}_t$, which
	estimate $\tilde{R}_t$ well, by taking the average of
	$W^x_{\pmb{s}_t}r^x_t$. In quantum algorithm is can be done with phase
	estimation.
	
	Finally, in lines~\ref{line:start}-\ref{line:end}, the algorithm
	updates the wights $W^x_{\pmb{s}_t}$ with the information of $\hat{R}_t$,
	which can be done in quantum computer with same gate complexity as classical
	as it is an arithmetic process.

	Our main result for the quantum algorithm is as follows:
	\begin{thm}[Query Complexity for quantum case]\label{thm:CompQuantum}
		There exists a quantum algorithm that can approximate the optimal
		coefficients in theorem~\ref{thm:ProbOS} with precision $\epsilon$
		with $\mathcal{O}\qty(\frac{\hat{c}}{\epsilon}T^2)$ queries of
		quantum basis classifiers $\hat{\mathcal{H}}$, and requires
		$N = \mathcal{O}\qty(\frac{\hat{c}^2}{\epsilon^2})$ data points as
		training sample.
	\end{thm}

	\begin{proof}
	At each iteration $t$, the querying process can be simply done with
	$\otimes_{i=1}^t \hat{\mathcal{H}}_i$ as shown above.

	Although the evaluating of $\hat{R}_t$, which is taking the average of all 
	branches, cannot be done in superposition as other steps, this can be 
	achieved with phase estimation. With an ancillary qubit
	$\ket{W^x_{\pmb{s}_t}}_M\ket{\pmb{s}_{t}}_{\mathcal{R}_{t}}\ket{0}$ it can 
	be converted to
	$$
		\sqrt{1-W^x_{\pmb{s}_t}/\hat{c}}
		\ket{W^x_{\pmb{s}_t}}_M\ket{\pmb{s}_t}_{\mathcal{R}_t}\ket{0}
		+
		\sqrt{W^x_{\pmb{s}_t}/\hat{c}}
		\ket{W^x_{\pmb{s}_t}}_M \ket{\pmb{s}_t}_{\mathcal{R}_t}\ket{1}
	$$
	if the last bit of $\mathcal{R}_t$ is in state $\ket{1}$, and left
	unchanged if the bit is in state $\ket{0}$ (Lemma 4 in \cite{Dervovic2018}).
	The $\hat{c}$ here is divided to make sure $\sqrt{1-W^x_{\pmb{s}_t}/
	\hat{c}}$ is always real, such that this operation is valid.

	If we start from the superposition of the whole sample, which is obtained
	from querying first $t$ classifiers $\otimes_{i=1}^t\hat{\mathcal{H}}_i$:
	$\frac{1}{\sqrt{N}}\sum_{x,\pmb{s}_t}
		\ket{x}_X \sqrt{q(\pmb{s}_t,x)}
		\ket{W^x_{\pmb{s}_t}}_M\ket{\pmb{s}_t}_{\mathcal{R}_t}
	$, this conditional operator gives
	\begin{equation}
		\begin{aligned}
			&
			\sqrt{1-\sum_{x,\pmb{s}_t}
			\frac{q(\pmb{s}_t|x) r^x_t W^x_{\pmb{s}_t}}
			{\hat{c} N}} |\varphi_0\rangle
			+
			\sqrt{\sum_{x,\pmb{s}_t}
			\frac{q(\pmb{s}_t|x) r^x_t W^x_{\pmb{s}_t}}
			{\hat{c} N}} |\varphi_1\rangle \\
			 & =
			\sqrt{1-\frac{\hat{R}_t}{\hat{c}}} \ket{\varphi_0}
			+
			\sqrt{\frac{\hat{R}_t}{\hat{c}}} \ket{\varphi_1} \ ,
		\end{aligned}
	\end{equation}
	where
	\begin{equation}
		\label{eq:Summation}
		\begin{aligned}
		|\varphi_0\rangle&
			=\frac{1}{\sqrt{N}}\sum_{x, \pmb{s}_t}
			\sqrt{q(\pmb{s}_t|x)}
			\frac{\sqrt{1 - r^x_t W^x_{\pmb{s}_t}/\hat{c}}}
			{\sqrt{1-\tilde{R}_t/\hat{c}}}
			\ket{x}\ket{\pmb{s}_t}
			|W^x_{\pmb{s}_t}\rangle
			\ket{0}\\
		|\varphi_1\rangle&
			=\frac{1}{\sqrt{N}}\sum_{x, \pmb{s}_t}
			\sqrt{q(\pmb{s}_t|x)}
			\frac{\sqrt{r^x_t W^x_{\pmb{s}_t}/\hat{c}}}
			{\sqrt{\tilde{R}_t/\hat{c}}}
			\ket{x}\ket{\pmb{s}_t}
			|W^x_{\pmb{s}_t}\rangle
			\ket{1} \ .
		\end{aligned}
	\end{equation}
	It is not hard to check that they are normalized.

	Rewrite this as
	$\sin(\theta_t)|\varphi_0\rangle+\cos(\theta_t)|\varphi_1\rangle$. The whole
	procedure is then simply rotate the initial state by $\theta_t$. Such
	operation can be used for phase estimation \cite{Shor1997}, where
	$\theta_t$ can be approximated with a constant probability with
	$\mathcal{O}(\frac{1}{\delta})$
	operations. The $\delta$ here is the precision of estimated
	$\hat{\theta}_t$, which is at most linear to the precision of
	$\hat{R}_t/\hat{c} = \cos^2(\theta_t)$.
	$\hat{R}_t$ can be easily calculated from the value of
	$\theta_t$, with precision $\hat{c}\delta$. That is, choosing
	$\delta = \frac{\epsilon}{\hat{c}}$, which means perform
	$\mathcal{O}\qty(\frac{\hat{c}}{\epsilon}) = \sqrt{N}$ times of the
	operation, is enough to estimate $\tilde{R}_t$ with precision $\epsilon$.
	Also, as discussed in the proof of theorem~\ref{thm:CompProb},
	$\abs{\hat{R}_t-\tilde{R}_t}\leq \epsilon$. Combine these precision
	together, it can be seen that our quantum algorithm achieves the same
	order of precision as the classical algorithm~\ref{alg:Classical}.

	Finally, The weight updating part of the classical AdaBoost algorithm can 
	be viewed as an operation on $\ket{\pmb{s}_{t}}_{\mathcal{R}_t}$ and
	$\ket{W^x_{\pmb{s}_t}}_M$ for each iteration $t$. With the updating
	rule~(\ref{eq:Updating}), $W^x_{\pmb{s}_t}$ can be easily obtained with
	$t-1$ iterations from $W^x_{\pmb{s}_1}\equiv 1$ with the information of
	$\tilde{R}_1,\tilde{R}_2, \cdots,\tilde{R}_{t-1}$ and each bit in
	$\mathcal{R}_t$ as control bit. Since in iteration $t+1$, the value
	$\tilde{R}_t$ is fixed for all inputs $x$, the division of $\tilde{R}_t$ can
	be applied simultaneously to all branches, regards the superposed nature of
	the weights. Since this is simply an arithmetic process, these operations
	can be implemented on quantum circuits with the same order of complexity as
	the classical circuits. Also there are no queries needed in this step. See
	appendix for the details of this implementation.

	However, the information of $\tilde{R}_1,\cdots,\tilde{R}_{t-1}$ are
	required for $t^\text{th}$ step, and hence above procedures have to be
	repeated for each $t$. Also, at $t^\text{th}$ step, the weights
	$W^x_{\pmb{s}_t}$ have to be adaptively updated from vary beginning as the
	measurements for phase estimation would disturb the quantum states.
	Therefore, at $t^\text{th}$ iteration, first $t$ basis classifiers need to
	be evaluated and hence the algorithm requires $\sum_{t=1}^T\mathcal{O}
	\qty(\frac{\hat{c}}{\epsilon} t) =
	\mathcal{O}\qty(\frac{\hat{c}}{\epsilon}T^2)$ queries in total. As we 
	choose $N=\mathcal{O}\qty(\frac{\epsilon^2}{\hat{c}^2})$, the query 
	complexity can be rewritten as $\mathcal{O}\qty(\sqrt{N}T^2)$.

	\end{proof}

%%%%%%%%%%%%%%%%%%%%%%%%%%%%%%%%%%%%%%
%\section{Quantum Learning }
%%%%%%%%%%%%%%%%%%%%%%%%%%%%%%%%%%%%%%
\textit{Quantum Learning---}Our quantum AdaBoost algorithm is also valid even if the inputs or basis classifiers are quantum, e.g., in the task of quantum template matching \cite{Sasaki2001} that aims to classify quantum states. In such an application, without loss of generality, we can consider the $t$-th binary classifier to be a two-outcome POVM $\qty{M^t_{-1},M^t_{+1}}$, where $M^t_{-1}+M^t_{+1} = \mathbb{I}$ are positive semidefinite operators. Denote the true label of $\rho_x$ by $y(\rho_x)\in\qty{-1,+1}$. It follows that $r^x_t = 0$ if the measurement outcome $h$ is equal to $y(\rho_x)$, and   $r^x_t = 1$ otherwise.
%As long as the true labels $y(\rho_x)\in\qty{-1,+1}$ of $\rho_x$ are provided for the training sample, the output of $M^t_h$ can be  relabeled to $r^x_t = 0$ if $h = y(\rho_x)$ and $r^x_t = 1$ otherwise, then 
Then the error probability of the $t$-th classifier on input state $\rho_x$ is given by $q_t(r^x_t =1\vert\rho_x)=\text{tr}\qty[\tilde{M}^t_1 \rho_x]$, where $\tilde{M}^t_{r^x_t}:=M^t_h$.

%One potential application of AdaBoost to quantum learning may comes from  t-design \cite{Zhu2016}. 
Now consider the task whose goal is to classify all pure states $\ket{\psi}\bra{\psi}$ into two groups, where each state is sampled from the Haar measure on the space of density operators. In general, a single POVM may not work well. However, if $T$ copies of the same state are available, then our quantum AdaBoost algorithm could be used to generate a strong POVM (aka. binary classifier) from individual weak POVM to improve the performance. 
%This task exactly fit in AdaBoost algorithm. 

Finally, we remark that the theory of t-design \cite{Zhu2016} could help in the implementation of our algorithm. 
In AdaBoost algorithm, the update rule requires taking expectation over the Haar measure on the space of density operators, which is hard to achieve in practical.
%over the uniform distribution of all pure states $$\int d\psi\ket{\psi}\bra{\psi}^{\otimes T},$$ where the integral is taking over the Harr measure on the space of density operators. However, such uniform distribution of all pure states may not be easy to obtain in practice. 
Nevertheless, according to the results from t-design \cite{Zhu2016}, it is possible to exactly simulate the uniform distribution of all pure states, in the sense that there exists a set of pure states $\ket{\psi_j}$ of size $K\approx\binom{d+T/2-1}{T/2}^2$ such that $\frac{1}{K}\sum_j \ket{\psi_j}\bra{\psi_j} =\int d\psi\ket{\psi}\bra{\psi}^{\otimes T}$. 
%This connection will help 
%, where there exists a set of pure states $\ket{\psi_j}$ of size $K\approx\binom{d+T/2-1}{T/2}^2$ that may be enough. 
%This implementation could be practical for experiments, and gives a potential application of boosting algorithms.
	
	% The error rate of
	% $\qty{\tilde{M}^i_1,M^i_{+1}}$ over all 
	% pure states is given by
	% $$ \int d\psi\ \text{tr}\qty[M^i_1\ket{\psi}\bra{\psi}]
	% =\tr\qty[M^i_1\qty(\int d\psi\ket{\psi}\bra{\psi})],$$
	% where the integral is taking over the Haar measure on the unit ball in the
	% input Hilbert space.
	% Furthermore, if $t$ copies of inputs are provided, and $t$ POVMs are
	% applied to each of the inputs, then the probability to produce a string
	% $\pmb{s}_t$ is
	% $$\text{tr}\qty[M^1_{s_1}\otimes\cdots\otimes M^t_{s_t}
	% \qty(\int d\psi\ket{\psi}\bra{\psi}^{\otimes t})].$$
	% As long as 
	% the weights $\alpha_t$ can be assigned to each measurement, the measurement
	% operators can be grouped according to AdaBoost model \ref{eq:AdaBoostModel},
	% which gives a new binary POVM that is optimal for given basis measures.

%%%%%%%%%%%%%%%%%%%%%%%%%%%%%%%%%%%%%%%%%%%%%%%%%
%\section{Discussion}
%%%%%%%%%%%%%%%%%%%%%%%%%%%%%%%%%%%%%%%%%%%%%%%%%

\textit{Discussion and Conclusion ---} In this article, we considered the conventional AdaBoost algorithm for binary classification tasks and extended it to probabilistic and quantum classifiers. The probabilistic extension is a straightforward analogy to quantum algorithms due to the probabilistic nature of quantum physics. As a result, the conventional AdaBoost algorithm can be perfectly recovered under the probabilistic extension. Furthermore, we showed that there exists a quantum algorithm that can optimize the AdaBoost model to the same precision as the classical algorithm with a quadratic speedup in terms of query complexity. Table~\ref{tab:AdaComplex} illustrates the complexities of AdaBoost algorithms in difference scenarios.
	\begin{table}[t]
		\centering
		\caption{Query Complexity of AdaBoost Models}
		\begin{tabular}{|c|c|c|}
			\hline
			AdaBoost Model & Type of Basis Classifier
			\footnote{D for deterministic classifier, P for probabilistic
			classifier, and Q for quantum classifier.} &
			Query Complexity                                                       \\
			\hline
			Conventional   & D                        & $\mathcal{O}(NT)$          \\
			Probabilistic  & D/P                      & $\mathcal{O}(NT)$          \\
			Quantum        & D/P/Q                    & $\mathcal{O}(\sqrt{N}T^2)$ \\
			\hline
		\end{tabular}
		\label{tab:AdaComplex}
	\end{table}
In realistic circumstances, $N\gg T$ holds in the application of  AdaBoost, hence, the query complexity in the quantum case performs better than that in the classical cases.  
%This can be seen from the fact that $\hat{c}=\max{W^x_{\pmb{s}_t}}= \prod_{t=1}^{T} \max\qty{\frac{1}{2\tilde{R}_t},\frac{1}{2\qty(1-\tilde{R}_t)}}$, and $N\propto \hat{c}^2$.

Boosting a collection of weak classifiers into a strong classifier is, in particular, suitable for building quantum learning machines because weak classifiers are easier to implement under current quantum hardware technology \cite{1809.06056}. Similar boosting ideas will most likely appear in designing noisy intermediate-scale quantum (NISQ) devices, and our quantum boosting algorithm can be employed to further improve the efficiency including those in \cite{1809.06056}.  

%	In case of quantum AdaBoost, since the output of each classifier could be any quantum state, it is clear that such framework is also good for quantum basis classifiers. However, similar to the classical case, little theory is known to guide the choice of quantum basis classifiers. Further numerical experiments might be needed to utilize our algorithm in practice.

\bibliographystyle{apsrev4-1}
\bibliography{QA}

\pagebreak
\begin{appendix}

%%%%%%%%%%%%%%%%%%%%%%%%%%%%%%%%%%%%%%%%%
%%%%%%%     Theoretical Optimum of AdaBoost Model        %%%%%%%%%%%%
%%%%%%%%%%%%%%%%%%%%%%%%%%%%%%%%%%%%%%%%%

\section{Proof of Theorem~\ref{thm:ProbOS}}\label{sec:TO}

	Here we perform our analysis for the probabilistic case, which can
	degenerate to the conventional AdaBoost if the outputs of
	classifiers are certain. Moreover, we assume a certain target label
	$y(x)\in\{+1,-1\}$ exists for all $x\in\Omega$, where $\Omega$ is the sample
	space of all possible inputs. Let $p(x)$ be the probability mass function
	defined on $\Omega$.

	The goal of AdaBoost is to find the optimal coefficients $\{\alpha_t\}$ of
	the linear model $g_T := \sum_{t=1}^T \alpha_t H_t$ based on the basis
	classifiers $H_t : \Omega \to \{+1, -1\}$ with the minimum exponential
	error, which is the average of $e^{-g_T(x) y(x)}$ over the joint
	distribution of inputs and classifiers. Here $\{H_t\}$ are random variables
	which yield the conditional probabilities $\mathbb{P}[H_t(x)=y(x)|x]$.
	Let $r^x_t:=\frac{1}{2}(1 - H_t(x)\cdot y(x))$, that is $r^x_t = 0$
	if $H_t(x) = y(x)$, and $r^x_t = 1$ otherwise. Then
	$\mathbb{P}[H_t(x)=y(x)|x]$ is fully determined by a conditional
	probability mass function $q_t(r^x_t|x)$.

	The exponential error as the cost function
	$C_T:={\mathbb{E}}[e^{g_T(x)y(x)}]$ yields
	\begin{equation}\label{eq:exponential}
		C_T=
		\sum_x p(x)
		\prod_{t=1}^T
		\left(
		\sum_{r_t^x =0,1}
		q_t(r_t^x|x)e^{-\alpha_t (-1)^{r_t^x}}
		\right)
	\end{equation}
	because $H_t(x)y(x) = (-1)^{r^x_t}$. In AdaBoost, the optimization problem
	is done by adding each term into $g_T(x)$ one by one with the optimal weight
	$\alpha_t$ at $t^\text{th}$ iteration. Let $C_t$ be the exponential error of
	the first $t$ terms of $g_T$

	Let $\pmb{s}_{t}\in\mathbb{B}^{t}$ be a binary string
	$s_1s_2\cdots s_t$. Let
	$w^x_{\pmb{s}_{t}}:=\prod_{i=1}^{t}e^{-\alpha_i(-1)^{s_i}}$, and let
	$q(\pmb{s}_{t},x) := p(x)\prod_{i=1}^{t} q_i(r^x_i=s_i|x)$. Then
	equation~(\ref{eq:exponential}) gives
	\begin{equation}
		\begin{aligned}
				& C_t=                   \\
				& \sum_{x\in\Omega,\pmb{s}_{t-1}}
			q(\pmb{s}_{t-1},x)w^x_{\pmb{s}_{t-1}}
			\left(
			\sum_{r^x_t =0,1}
			q_t(r^x_t|x)e^{-\alpha_t (-1)^{r^x_t}}
			\right).
		\end{aligned}
	\end{equation}
	This is a convex function respect to $\alpha_t$, and an unique solution to the problem exists at the extreme. Taking its derivative to $0$ gives
	\begin{multline}
		\sum_{x,\pmb{s}_{t-1}} q(\pmb{s}_{t-1},x)w^x_{\pmb{s}_{t-1}}
		\left(
		-q_t(0|x)e^{-\alpha_t}
		+
		q_t(1|x)e^{\alpha_t}
		\right) \\
		=0
	\end{multline}
	and hence
	\begin{equation}
		\begin{aligned}
				& e^{2\alpha_t} &
			\sum_{x,\pmb{s}_{t-1}}
			q(\pmb{s}_{t-1},x)w^x_{\pmb{s}_{t-1}} q_t(1|x) \\
			= &               & \sum_{x,\pmb{s}_{t-1}}
			q(\pmb{s}_{t-1},x)w^x_{\pmb{s}_{t-1}} q_t(0|x) .
		\end{aligned}
	\end{equation}

	That is
	\begin{equation}
		\label{eq:ProbAlpha}
		\alpha_t
		=
		\frac{1}{2}
		\ln
		\frac
		{\sum_{x,\pmb{s}_{t-1}}
		q(\pmb{s}_{t-1},x)w^x_{\pmb{s}_{t-1}} q_t(0|x)}
		{\sum_{x,\pmb{s}_{t-1}}
		q(\pmb{s}_{t-1},x)w^x_{\pmb{s}_{t-1}} q_t(1|x)}.
	\end{equation}

	Let
	\begin{equation}
		\tilde{R}_t:=
		\frac
		{\sum_{x,\pmb{s}_{t-1}}
		q(\pmb{s}_{t-1},x)w^x_{\pmb{s}_{t-1}} q_t(1|x)}
		{\sum_{x,\pmb{s}_{t-1}}
		q(\pmb{s}_{t-1},x)w^x_{\pmb{s}_{t-1}}} \ .
	\end{equation}
	Then the optimal weight of each iteration is
	\begin{equation}
		\alpha_t
		=
		\frac{1}{2}
		\ln
		\left(
		\frac{1-\tilde{R}_t}{\tilde{R}_t}
		\right).
	\end{equation}

	In the following, we demonstrate that the optimal weight $\{\alpha_t\}$
	can be adaptively obtained. When $t=1$, initialize
	$w^x_{\pmb{s}_0}\equiv 1$. Thus
	\begin{equation}
		\begin{aligned}
			\tilde{R}_1
				& =
			\frac
			{\sum_{x}p(x) q_1(1|x)}
			{\sum_{x}p(x)}        \\
				& =
			\sum_{x}p(x) q_1(1|x) \\
				& =
			\mathbb{E}_{p\times q_1}[r^x_1]
		\end{aligned}
	\end{equation}
	which is exactly the generalization error of $H_1$.\\

	Let $Z_t = \sum_{x,\pmb{s}_{t-1}}q(\pmb{s}_{t-1},x)
	w^x_{\pmb{s}_{t-1}}$ be the $t^\text{th}$ normalization factor. Then
	\begin{equation}
		\label{eq:rt}
		\begin{aligned}
			\tilde{R}_t & =
			\sum_{x,\pmb{s}_{t-1}} q(\pmb{s}_{t-1},x)q_t(1|x)
			\frac{w^x_{\pmb{s}_{t-1}}}{Z_t}                          \\
				& =\sum_{x,\pmb{s}_{t-1}}q(\pmb{s}_{t-1},x)\sum_{r^x_t=0,1}
			q_t(r^x_t|x) \frac{w^x_{\pmb{s}_{t-1}}}{Z_t} r^x_t       \\
				& =	\mathbb{E}_q
			\left[
			\frac{w^x_{\pmb{s}_{t-1}}}{Z_t}r^x_t
			\right].
		\end{aligned}
	\end{equation}

	By definition $w^x_{\pmb{s}_{t}}=w^x_{\pmb{s}_{t-1} r_t} = w^x_{\pmb{s}_{t-1}}
		e^{-\alpha_t(-1)^{r_t}}$.  Therefore
	\begin{equation}
		\begin{aligned}
			Z_{t+1}
				& = \sum_{x,\pmb{s}_{t}}q(\pmb{s}_{t},x)w^x_{\pmb{s}_{t}} \\
				& = \sum_{x,\pmb{s}_{t-1}}q(\pmb{s}_{t-1}|x)w^x_{\pmb{s}_{t-1}}
			(q_t(0|x)e^{-\alpha_t}+q_t(1|x)e^{\alpha_t})       		   \\
				& =e^{-\alpha_t} (Z_t \tilde{R}_t e^{2\alpha_t} +
				Z_t(1-\tilde{R}_t))                        \\
				& =e^{-\alpha_t} (Z_t \tilde{R}_t
				\frac{1 -\tilde{R}_t}{\tilde{R}_t} + Z_t(1-\tilde{R}_t))  \\
				& =2e^{-\alpha_t}(1-\tilde{R}_t)Z_t.
		\end{aligned}
	\end{equation}

	Let %$\frac{w^x_{\pmb{s}_t}}{Z_{t+1}}$ as $W^x_{\pmb{s}_t}$.
	\begin{equation}
		\label{eq:r1}
		\begin{aligned}
				& W^x_{\pmb{s}_{t+1}} := \frac{w^x_{\pmb{s}_{t}}}{Z_{t+1}}            \\
			= & \frac{w^x_{\pmb{s}_{t-1}}}{Z_{t+1}} e^{-\alpha_t(-1)^{r_t}}      \\
			= & W^x_{\pmb{s}_{t}} \frac{1}{2(1-\tilde{R}_t)}e^{\alpha_t(1-(-1)^{r_t})}.
		\end{aligned}
	\end{equation}

	Therefore, all the values of $W^x_{\pmb{s}_{t}}$ can be obtained by
	iterating with the information of $\tilde{R}_t$ . It is not hard to
	check that (\ref{eq:r1}) is equivalent to the updating
	rule~ \eqref{eq:Updating}. These values again yield
	$\tilde{R}_{t+1}={\mathbb{E}}[W^x_{\pmb{s}_{t+1}} r^x_{t+1}]$ for next
	iteration, and therefore every $\alpha_t$ could be determined
	analytically in this manner.

	%%%%%%%%%%%%%%%%%%%%%%%%%%%%%%%%%%%%%%
	%%%%%%%           Classical AdaBoost Algorithm        %%%%%%%%%%%%
	%%%%%%%%%%%%%%%%%%%%%%%%%%%%%%%%%%%%%%

\section{Proof of Theorem~\ref{thm:CompProb}}\label{sec:ClassA}

	In the section~\hyperref[sec:TO]{Proof of Theorem~\ref{thm:ProbOS}}, a
	theoretical optimal solution to the AdaBoost Model is derived. However,
	in practice, the underlying distribution of inputs is unknown, and
	therefore the values of $q(\pmb{s}_{t}|x)$ are impossible to
	be evaluated. Also, usually the training algorithm cannot cover the
	whole sample space (otherwise the explicit relationship between inputs
	and output are known, and machine learning is unnecessary).

	Similar to other machine learning tasks, this problem is solved by
	sampling.  Clearly, with a underlying distribution $D$ on the sample
	space $\Omega$, each $r^x_t$ can be viewed as a random variable on the
	sample space. This can be done with an interesting result derived from
	Hoeffding's inequality.
	% \begin{thm}[Hoeffding's Inequality]
	% 	Let $\bar{X}$ be the average of N independent random variables, i.e.
	% 	$\bar{X} = \frac{1}{N}\qty(X_1+\cdots+X_N)$, then
	% 	\begin{equation}
	% 		\mathbb{P} \qty[
	% 			\bar{X}-\mathbb{E}(\bar{X})|\geq \epsilon
	% 		] \leq
	% 		2 \exp(-\frac{2N^2 \epsilon^2}{\sum^N_{i=1}c_i^2})
	% 	\end{equation}
	% 	where $c_i :=\sup X_i-\inf X_i$, and $\mathbb{E}(\bar{X})$ is the
	% 	expected value of $\bar{X}$.
	% \end{thm}
	% Suppose the random variables $X_i$ are all the same, then
	% $c_i$ are the same for all $i$, and rewrite them as $c$.
	% Clearly, $\mathbb{E}[\bar{X}]=\mathbb{E}[X]$. Then we have
	\begin{thm}[Hoeffding's inequality]
		\label{thm:MDinequal}
		If a sample $S$ of size $N$ is drawn from a distribution $D$ on a
		sample space $\Omega$, then given a random variable $X$ on $\Omega$
		and any positive number $\epsilon>0$
		\begin{equation}
			\mathbb{P}
			\qty[
				\abs{\frac{1}{N}\sum_{x\in S} X(x) -
				\underset{D}{\mathbb{E}}[X]}
				\geq \epsilon
			]\leq 2\exp(-\frac{2N\epsilon^2}{\hat{c}^2})
		\end{equation}
		where $c := \qty[\sup{X}-\inf{X}]$.
	\end{thm}

	%Moreover, according to the definition of error, and the
	%fact that $\max{r(x)}-\min{r(x)}=1$ this theorem gives the corollary
	%
	%\begin{cor}
	%	For any hypothesis $h$, its empirical error on a sample $S$ with
	%	size $N$ drawn from the underlying distribution $D$ approximates
	%	its generalization error well with high probability. That is
	%	\begin{equation}
	%		\mathbb{P}
	%		\left(
	%			|R(h)-\hat{R}(h)|
	%			\geq
	%			\epsilon
	%		\right)
	%		\leq
	%		2e^{-2\epsilon^2N}
	%	\end{equation}
	%\end{cor}

	The key point here is that, though
	$\underset{x\sim D}{\mathbb{E}}[X(x)]$ cannot be evaluated in practice,
	$\frac{1}{N}\sum_{x\in S} X(x)$ is computable, and it approximates
	$\underset{x\sim D}{\mathbb{E}}[X(x)]$ well when $N = |S|$ is large.

	According to equation~(\ref{eq:rt}),
	$\tilde{R}_t =
	\mathbb{E}_q
	\left[
		W^x_{\pmb{s}_t}r^x_t
	\right].$
	For a sample $S$ of pairs
	$(x, \pmb{s}_{t})$ drawn from the distribution $q(\pmb{s}_t,x)$, let
	\begin{equation}
		\label{eq:ApproxAda}
		\hat{R}_t
		:=
		\frac{1}{N}\sum_{x\in S}\left[W^x_{\pmb{s}_{t}} r^x_t\right] .
	\end{equation}

	Then theorem \ref{thm:MDinequal} shows that
	\begin{equation}
		\mathbb{P}
		\left(
		|\tilde{R}_t-\hat{R}_t|
		\geq
		\epsilon
		\right)
		\leq
		2e^{-2\epsilon^2N/c_t^2} \ ,
	\end{equation}
	where $N$ is the size of $S$ and
	$c_t := \max (W^x_{\pmb{s}_{t}} r^x_t)-\min (W^x_{\pmb{s}_{t}} r^x_t)$.

	To be noticed, the value of $W^x_{\pmb{s}_t}$ is derived with iteration
	according to equation~(\ref{eq:r1}). Since $W^x_{\pmb{s}_1}\equiv 1$,
	$W^{x}_{\pmb{s}_t}$ is always positive, which means
	$\min(W^x_{s_t} r^x_t)$ is always non-negative as well. Further,
	$\max{W^x_{s_t}r^x_t}\leq \max{W^x_{s_t}}:= \hat{c}_t$.

	Therefore, for a target precision $\epsilon$ of $\hat{R}_t$, a sample
	with size $N = \mathcal{O}\left(\frac{\hat{c}_t^2}{\epsilon^2} \right)$
	is good enough to achieve the goal with a constant probability.
	Nevertheless, the size of sample have to be determined before hand; and
	hence we should choose
	\begin{equation}
		\label{eq:SampleSize}
	N = \mathcal{O}\qty(\frac{\hat{c}^2}{\epsilon^2}) \ ,
	\end{equation}
	where $\hat{c}=\max\qty{\hat{c}_t}$.

	\begin{rem}
		However, $\hat{c}$ might not be small when $T$ is large, which
		indicates that AdaBoost may not be good if the model does not
		converge fast with the number of classifiers used. These might be
		improved by other boosting algorithms, e.g. LogitBoost, Gradient
		Boosting, XGBoosting.
	\end{rem}

	As long as we obtain a sample $S$ of size $N$, according to
	theorem~\ref{thm:MDinequal}, the algorithm~\ref{alg:Classical} 
	approximates $\tilde{R}_t$ well. This algorithm evaluates each data 
	$x\in S$ for each classifier $H_t$, and therefore requires
	$\mathcal{O}(NT)$ queries.

%%%%%%%%%%%%%%%%%%%%%%%%%%%%%%%%%%%%%%
%%%%%%%       Quantum Simulation of Classical Process    %%%%%%%%%%%%
%%%%%%%%%%%%%%%%%%%%%%%%%%%%%%%%%%%%%%
\section{Quantum Simulation of Classical Process}
	This section reviews some results from Kitaev's paper
	\cite{Kitaev1995} that simulate classical Boolean circuits with
	quantum circuits. For convenient, without loose of generality the
	classical registers are denoted with Dirac notations here.

	According to lemma~1 and 7 in \cite{Kitaev1995}, if a function
	$F:\mathbb{B}^n\to\mathbb{B}^m$ can be computed with $L$ Boolean
	operations $g\in$ a basis $\mathcal{B}$, which is a small set of
	Boolean operations, then it can be computed with $2L+m$ operations in
	the basis $\mathcal{B}_\tau$. The basis $\mathcal{B}_\tau$ is defined
	in a way that, for each $g\in\mathcal{B}:\ket{x}_X\to\ket{g(x)}_B$,
	there is a
	$ g_\tau\in\mathcal{B}_\tau:\ket{x}_X\ket{v}_B
	\to \ket{x}_X\ket{v\oplus g(x)}_B \ . $
	Also the operation to copy a state
	\begin{equation}
		\label{eq:Tau}
		\tau_{A,B}:\ket{x}_A\ket{v}_B\to\ket{x}_A\ket{v\oplus x}_B
	\end{equation}
	(which is indeed a CNOT gate) have to be included into
	$\mathcal{B}_\tau$.

	Furthermore, we say a circuit computes a Boolean function $F$, if it
	converts $\ket{x}_X\to\ket{F(x)}_B$. With
	basis $\mathcal{B}_\tau$, this is computation is performed as
	$\ket{x}_X\ket{0}_B \to \ket{x}_X\ket{F(x)}_B$.

	% If we choose the basis $\mathcal{B}$ to be
	% $\qty{\text{NOT}(\neg),\text{AND}(\wedge)}$, then clearly,
	% $\mathcal{B}_\tau$ is $\qty{\text{NOT},\text{CNOT},\text{Toffoli}}$.
	% Note that, the latter one is easily achievable with quantum operations.

	However, one may only need partial information about the output $F(x)$.
	Classically, it is free to readout part of the the output bits and drop
	the rest. Nevertheless, in quantum computation, dropping those
	``garbage'' bits ($\ket{\text{gar}(x)}$) would destroy the quantum
	state if they are in superposition. But as shown above,
	$\ket{x}_X\ket{0}_B\to\ket{x}_X\ket{F(x)}_B$ can be
	constructed with $2L+m$ reversible gates. Divide register $B$ into two
	parts $\qty(B_1,B_2)$, and then $\ket{F(x)}_B$ is
	$\ket{f(x)}_{B_1}\otimes\ket{\text{gar}(x)}_{B_2}$ (it is always
	separable as the initial states are all tensor product states), above
	process is then
	$$
	\ket{x}_X\ket{0}_{B_1}\ket{0}_{B_2}\to
	\ket{x}_X\ket{f(x)}_{B_1}\ket{\text{gar}(x)}_{B_2} .
	$$ By repeating this process on an extra register $B^\prime=
	\qty(B_1^\prime,B_2^\prime)$, the process
	$$
	\ket{x}_X\ket{0}_{B_1}\ket{0}_{B_2}
	\ket{0}_{B_1^\prime}\ket{0}_{B_2^\prime}
	$$
	$$
	\to \ket{x}_X\ket{H(x)}_{B_1}\ket{\text{gar}(x)}_{B_2}
	\ket{f(x)}_{B_1^\prime} \ket{\text{gar}(x)}_{B_2^\prime}
	$$ can be
	constructed.

	If the input state is on quantum registers and it is in superposition
	$$
	\qty(\frac{1}{\sqrt{N}}\sum_x\ket{x}_X)\ket{0}_{B_1}\ket{0}_{B_2}
	\ket{0}_{B_1^\prime}\ket{0}_{B_2^\prime},
	$$
	this process will give
	$$
	\frac{1}{\sqrt{N}}\sum_x\ket{x}_X
	\ket{f(x)}_{B_1}\ket{\text{gar}(x)}_{B_2}
	\ket{f(x)}_{B_1^\prime}\ket{\text{gar}(x)}_{B_2^\prime} .
	$$
	Then the pairwise operation $\tau_{B_2^\prime,B_2}$ \eqref{eq:Tau} is
	performed between the ``garbage'' states on $B_2$ and $B_2^\prime$,
	which gives
	$$
	\frac{1}{\sqrt{N}}\sum_x\ket{x}_X
	\ket{f(x)}_{B_1}\ket{\text{gar}(x)\oplus\text{gar}(x)}_{B_2}
	\ket{f(x)}_{B_1^\prime}\ket{\text{gar}(x)}_{B_2^\prime}
	$$
	$$
	= \frac{1}{\sqrt{N}}\sum_x\ket{x}_X
	\ket{f(x)}_{B_1}\ket{0}_{B_2}
	\ket{f(x)}_{B_1^\prime}\ket{\text{gar}(x)}_{B_2^\prime} .
	$$

	Finally, the original process is
	performed again on $B^\prime$ which ends up at
	$$
	\qty(\frac{1}{\sqrt{N}}\sum_x\ket{x}_X
	\ket{f(x)}_{B_1})\ket{0}_{B_2}
	\ket{0}_{B_1^\prime}\ket{0}_{B_2^\prime} .
	$$
	Since the appending registers
	$B_2, B_1^\prime, B_2^\prime$ are all end up at $\ket{0}$, it
	is free to drop them after computation.

	In summary, the process $\ket{x}_X\ket{0}\to\ket{x}_X\ket{f(x)}$ can be
	achieved with $\mathcal{O}(L)$ quantum gates even for computation in
	superposition. This fact indicates that each arithmetic part in our
	quantum algorithm can be performed with the same complexity of the
	classical algorithm. Since all ancillary registers always start and end
	at $\ket{0}$, they are neglected in our notation for simplicity.

	Note that, although above result is only valid for Boolean functions,
	as how modern computers work, these Boolean operations are indeed
	universal. In case people want to deal with real numbers on computers,
	those values have to be encoded into binary strings up to some
	precision.

	\begin{exmp}
		\label{exmp:Arithmestic}
		The updating rule~(\ref{eq:Updating}) is purely
		arithmetic.	This can be viewed as repeating controlled operation
		$\mathcal{U}$ on a register $M$, encoding a numerical value $\xi$
		in terms of binary strings $\ket{\xi}$ up to some precision. Each
		application of $\mathcal{U}$ is controlled by each qubit of the
		string $\ket{\pmb{s}_{t}}_{\mathcal{R}_t}:=
		\ket{s_1}\otimes\cdots\otimes\ket{s_t}$. More precisely,
		$$
			\mathcal{U}_i =
			U_0\otimes \ket{s_i=0}\bra{s_i=0}_i
			+ U_1\otimes \ket{s_i=1}\bra{s_i=1}_i \ ,
		$$
		where $U_0\ket{\xi}_M = \ket{\frac{\xi}{2(1-R_t)}}_M$;
		$U_1\ket{\xi}_M = \ket{\frac{\xi}{2R_t}}_M$. As a result,
		lines~\ref{line:start}-\ref{line:end} in
		algorithm~\ref{alg:Classical} can be performed in quantum circuits
		with the same order of gates as classical circuit. Additionally,
		this can be done in superposition for all $x$, and hence the
		``for'' loop in classical algorithm can be done in one shot.
	\end{exmp}

	Similarly, another step for phase estimation in our algorithm can be
	done with this method.
	\begin{exmp}
		\label{exmp:Rotate}
		There exists an operation $\mathcal{Q}_t$ such that for
		$\xi\in\qty[0,\hat{c}]$,
		\begin{equation}
			\label{eq:RotateOperator}
			\mathcal{Q}_t\ket{\xi}_M\ket{0}
			=\ket{\xi}_M(\sqrt{1-\frac{\xi}{\hat{c}}}\ket{0} +
			\sqrt{\frac{\xi}{\hat{c}}}\ket{1} \ .
		\end{equation}
		The requirement of $\xi\in\qty[0,\hat{c}]$ is presented to make sure
		$\cos^{-1}\sqrt{\frac{\xi}{\hat{c}}}$ is a real number, and 
		therefore, the state
		$\ket{\cos^{-1}\sqrt{\frac{\xi}{\hat{c}}}}_{\text{anc}}$
		can be constructed on an ancillary register ``anc''. Here
		$\ket{A}_{\text{anc}} = \ket{a_1}\ket{a_2}\cdots\ket{a_m}$, where
		$0.a_1a_2\cdots a_m$ is the binary representation of the real
		number $A$ up to some precision. The process to compute
		$\ket{\xi}_M\ket{0}_{\text{anc}}\to
		\ket{\xi}_M\ket{\cos^{-1}\sqrt{\frac{\xi}{\hat{c}}}}$ is arithmetic.
		By further appending an additional qubit $\ket{0}$ to the system,
		an operation can be constructed as lemma 4 in \cite{Dervovic2018},
		such that it converts $\ket{\xi}_M\ket{A=
		\cos^{-1}\sqrt{\frac{\xi}{\hat{c}_t}}}_{\text{anc}}\ket{0}$ to
		$$\ket{\xi}_M\ket{A}_{\text{anc}}(\sin(A)\ket{0}+\cos(A)\ket{1})$$
		$$=\ket{\xi}_M\ket{A}_{\text{anc}}
		(\sqrt{1-\frac{\xi}{\hat{c}}}\ket{0}
		+ \sqrt{\frac{\xi}{\hat{c}_t}}\ket{1}).$$
		Finally the register $\ket{A}_\text{anc}$ can be cleared and
		dropped with the garbage dropping technique above. This whole 
		process is exactly the operation $\mathcal{Q}_t$.
	\end{exmp}

	The operations in these examples would be useful in next section.

%%%%%%%%%%%%%%%%%%%%%%%%%%%%%%%%%%%%%%
%%%%%%%       Quantum AdaBoost Algorithm            %%%%%%%%%%%%
%%%%%%%%%%%%%%%%%%%%%%%%%%%%%%%%%%%%%%
\section{Proof of Theorem~\ref{thm:CompQuantum}}\label{sec:QAda}

	In the Quantum AdaBoost Algorithm, the computation other then the
	average of $r^x_t W^x_{\pmb{s}_{t}}$ can be performed in parallel on the
	whole sample. That is, for every initial state $\ket{x}_X$, where $x$
	is the data points of a sample $S$ of size $N$ drawn from the sample
	space $\Omega$, the classical algorithm outputs
	$|r^x_t W^x_{\pmb{s}_{t}}\rangle_{M}$ to the register $M$, which
	encoding the numerical value of $r^x_t W^x_{\pmb{s}_{t}}$. Note that
	$\ket{r^x_t W^x_{\pmb{s}_{t}}}_M$ here is the state corresponding to the
	binary value of $r^x_t W^x_{\pmb{s}_{t}}$, as how modern computer saves
	numerical values. With this property, the AdaBoost algorithm can be
	performed by following adaptive procedure:

	At $t^\text{th}$ iteration, given the classical information of
	$\tilde{R}_1, \cdots,\tilde{R}_{t-1}$ (where $\tilde{R}_t$ can be
	obtained in $t^\text{th}$ iteration), initialize the state of three
	registers $X$, $M$, $\mathcal{R}_t$ as
	\begin{equation}
		\label{eq:Init}
		\frac{1}{\sqrt{N}}
		\sum_{x\in S}
		\ket{x}_X\otimes\ket{W^x_{\pmb{s}_1}\equiv 1}_M
		\otimes\ket{0}_{\mathcal{R}_t}^{\otimes t} ,
	\end{equation}
	with access to the quantum oracle
	$\hat{\mathcal{H}}_1\otimes\cdots\otimes\hat{\mathcal{H}}_t$ defined in
	\eqref{eq:QQuery}, one can obtain
	\begin{equation}
		\label{eq:Query}
		\frac{1}{\sqrt{N}}
		\sum_{x\in S}\sum_{\pmb{s}_t}
		\ket{x}_X\otimes\ket{W^x_{\pmb{s}_1}}_M\otimes
		\sqrt{q(\pmb{s}_t\vert x)}\ket{\pmb{s}_t}_{\mathcal{R}_t} .
	\end{equation}
	With the classical information of $\tilde{R}_1,\cdots,\tilde{R}_{t-1}$,
	one can update the register $M$ with the updating
	rule~(\ref{eq:Updating}) (which is a classical arithmetic process shown
	in example~\ref{exmp:Arithmestic}) to the state
	\begin{equation}\label{eq:QuantumClassical}
		\frac{1}{\sqrt{N}}
		\sum_{x\in S}\sum_{\pmb{s}_{t}}
		\ket{x}_X\otimes\ket{W^x_{\pmb{s}_{t}}}_M\otimes
		\sqrt{q(\pmb{s}_{t}\vert x)}\ket{\pmb{s}_{t}}_{\mathcal{R}_t} \ .
	\end{equation}
	Compose the whole arithmetic process that converts \eqref{eq:Init}
	to \eqref{eq:QuantumClassical} and rewrite it as $\mathcal{A}_t$:
	\begin{equation}
		\label{eq:Arithmatic}
		\begin{aligned}
		&\mathcal{A}_t
		\frac{1}{\sqrt{N}} \sum_{x\in S}
		\ket{x}_X\otimes\ket{W^x_{\pmb{s}_1}}_M\otimes
		\ket{0}_{\mathcal{R}_t}\\
		=&
		\frac{1}{\sqrt{N}}
		\sum_{x\in S}\sum_{\pmb{s}_{t}}\sqrt{q(\pmb{s}_{t}\vert x)}
		\ket{x}_X\otimes\ket{W^x_{\pmb{s}_{t}}}_M\otimes
		\ket{\pmb{s}_{t}}_{\mathcal{R}_t} \ ,
		\end{aligned}
	\end{equation}

	With an extra working register, apply the operation in
	example~\ref{exmp:Rotate} to the final state in
	\eqref{eq:Arithmatic}
	\begin{equation}
		\label{eq:PrepPhaseEstimation}
		\begin{aligned}
				& \mathcal{Q}_t \mathcal{A}_t
			\frac{1}{\sqrt{N}}
			\sum_{x\in S^\prime}
			\ket{x}_X\ket{W^x_{\pmb{s}_1}}_M
			\ket{0}_{\mathcal{R}_t}
			\ket{0} =  \\
				&
			\sqrt{
			1-
			\sum_{x, \pmb{s}_t}
			\frac{q(\pmb{s}_t|x) r^x_t W^x_{\pmb{s}_t}}{\hat{c} N}
			}
			|\varphi_0\rangle\ket{0}
			+
			\sqrt{
			\sum_{x, \pmb{s}_t}
			\frac{q(\pmb{s}_t|x) r^x_t W^x_{\pmb{s}_t}}{\hat{c} N}
			}
			|\varphi_1\rangle\ket{1} ,
		\end{aligned}
	\end{equation}
	where
	\begin{equation}
		\begin{aligned}
		\ket{\varphi_0}&
			=\frac{1}{\sqrt{N}}\sum_{x, \pmb{s}_t}
			\sqrt{q(\pmb{s}_t|x)}
			\frac{\sqrt{1 - r^x_t W^x_{\pmb{s}_t}/\hat{c}}}
			{\sqrt{1-\hat{R}_t/\hat{c}}}
			\ket{x}_X
			\ket{W^x_{\pmb{s}_t}}_M
			\ket{\pmb{s}_t}_{\mathcal{R}_t}\\
		\ket{\varphi_1}&
			=\frac{1}{\sqrt{N}}\sum_{x, \pmb{s}_t}
			\sqrt{q(\pmb{s}_t|x)}
			\frac{\sqrt{r^x_t W^x_{\pmb{s}_t}/\hat{c}}}
			{\sqrt{\hat{R}_t/\hat{c}}}
			\ket{x}_X
			\ket{W^x_{\pmb{s}_t}}_M
			\ket{\pmb{s}_t}_{\mathcal{R}_t} .
		\end{aligned}
	\end{equation}
	Note that for each $x$, $\sum_{\pmb{s}_t} q(\pmb{s}_t|x) = 1$.

	According to the definition in equation~\eqref{eq:ApproxAda}, the result of \eqref{eq:PrepPhaseEstimation} is indeed
	\begin{equation}\label{eq:QAda}
		\sqrt{1-\frac{\hat{R}_t}{\hat{c}}}|\varphi_0\rangle\ket{0}
		+
		\sqrt{\frac{\hat{R}_t}{\hat{c}}}|\varphi_1\rangle\ket{1} .
	\end{equation}
	This can be rewrite as
	\begin{equation}\label{eq:Rotate}
		\ket{\psi_0} :=
		\sin(\theta_t)\ket{\varphi_0}\ket{0}
		+
		\cos(\theta_t)\ket{\varphi_1}\ket{1} ,
	\end{equation}
	which performs a rotation of angle $\theta_t$.

	Let $|\psi_1\rangle := \cos(\theta_t)|\varphi_0\rangle\ket{0}
	- \sin(\theta_t)|\varphi_1\rangle\ket{1}$. After a Pauli-$\mathcal{Z}$
	operation is performed on the last register of $\ket{\psi_0}$, it is
	transformed to
	\begin{equation}
		\label{eq:Rotated}
		\begin{aligned}
			&  & \sin(\theta_t)
			[\sin(\theta_t)|\psi_0\rangle+\cos(\theta_t)|\psi_1\rangle] \\
			&- & \cos(\theta_t)
			[\cos(\theta_t)|\psi_0\rangle-\sin(\theta_t)|\psi_1\rangle] \\
			=& &
			\cos(2\theta_t)|\psi_0\rangle + \sin(2\theta_t)|\psi_1\rangle .
		\end{aligned}
	\end{equation}

	Let $\mathcal{G}_t:=\mathcal{Q}_t\mathcal{A}_t$. Apply the inverse
	operation $\mathcal{G}_t^\dagger$ to \eqref{eq:Rotated}, so that
	$\ket{\psi_0}$ is mapped back to the initial state \eqref{eq:Init}.
	Note that, $|\psi_1\rangle$ is orthogonal to $|\psi_0\rangle$ and our
	operation is unitary. Therefore, if an operation $U_\perp$ only
	inverse the amplitude of the every state perpendicular to the initial
	state (analogy to the diffusion operator in Grover's algorithm
	\cite{Grover1996}) is applied and the operation $\mathcal{G}_t$ is
	performed again, $\ket{\psi_0}$ would be left unchanged. This procedure
	gives
	\begin{equation}
		\begin{aligned}
			&  & \cos(2\theta_t)|\psi_0\rangle
				- \sin(2\theta_t)|\psi_1\rangle \\
			= && \cos(2\theta_t)
			[
				\sin(\theta_t)|\varphi_0\rangle\ket{0}
				+
				\cos(\theta_t)|\varphi_1\rangle\ket{1}
			]                                                                 \\
				& - & \sin(2\theta_t)
			[
				\cos(\theta_t)|\varphi_0\rangle\ket{0}
				-
				\sin(\theta_t)|\varphi_1\rangle\ket{1}
			]                                                                 \\
			= && \cos(3\theta_t)|\varphi_0\rangle\ket{0}
			+
			\sin(3\theta_t)|\varphi_1\rangle\ket{1} .
		\end{aligned}
	\end{equation}

	In conclusion,
	$(\mathcal{G}_tU_\perp\mathcal{G}_t^\dagger\mathcal{Z})^k\mathcal{G}_t$
	converts the initial state to
	\begin{equation}
		\label{eq:Phasek}
	\cos((2k+1)\theta_t)|\varphi_0\rangle\ket{0}
	+
	\sin((2k+1)\theta_t)|\varphi_1\rangle\ket{1} \ .
	\end{equation}
	Such operation provides the possibility to estimate $\theta_t$ with
	phase estimation algorithm.

	To fairly compare the query complexities, we want to constrain the
	results from both classical and quantum algorithm to the same
	precision. In order to approximate $\tilde{R}_t$ with the target
	precision $\mathcal{O}(\epsilon)$, the phase estimation algorithm have
	to estimate $\hat{R}_t = \cos^2(\theta_t)\hat{c}$ with precision
	$\epsilon$, and as shown in \eqref{eq:SampleSize}, a sample of size
	$N=\mathcal{O}\qty(\frac{\hat{c}^2}{\epsilon^2})$ is enough to estimate
	each $\tilde{R}_t$ with $\hat{R}_t$ with precision $\epsilon$.

	In the $t^\text{th}$ step of our quantum algorithm, by choosing number 
	of iterations in \eqref{eq:Phasek} to be
	$k=\mathcal{O}\qty(\frac{1}{\delta})$, the phase
	estimation process could read out the value of $\hat{\theta}_t$, such 
	that $\abs{\theta_t-\hat{\theta}_t}\leq\delta$.

	In order to estimates $\tilde{R}_t$ with the same precision as
	the classical algorithm, we need to bound $\hat{\theta}$ to make sure
	$$|\hat{c} \cos^2(\hat{\theta}_t)-\hat{R}_t|\leq \epsilon \ .$$
	This can be done by choose a proper $\delta$. Then, the task of our 
	analysis is to bound the value of $\delta$ in terms of $\epsilon$ and 
	$\hat{c}$ as in the classical case.

	Let $\hat{\epsilon}:= \hat{c} \cos^2(\hat{\theta}_t)-\hat{R}_t$, then
	$\abs{\hat{\theta}_t-\theta_t} \leq \delta$ gives
	$\abs{ \cos^{-1}(\sqrt{\frac{\hat{R}_t+\hat{\epsilon}}{\hat{c}}})
	- \cos^{-1}(\sqrt{\frac{\hat{R}_t}{\hat{c}}}) } \leq \delta.$
	Since $\abs{\hat{\epsilon}}\leq \epsilon$, $\hat{\epsilon}$ is a small
	number, and hence
	\begin{equation}
		\begin{aligned}
		&\abs{
			\cos^{-1}(\sqrt{\frac{\hat{R}_t+\hat{\epsilon}}{\hat{c}}})
			-
			\cos^{-1}(\sqrt{\frac{\hat{R}_t}{\hat{c}}})
		}\\
		\approx & \abs{\frac{\hat{\epsilon}}{\hat{c}}
		\frac{d}{dx}\cos^{-1}(\sqrt{x})\vert_{x=\hat{R}_t/\hat{c}}} \ .
		\end{aligned}
	\end{equation}
	When $\frac{\hat{R}_t}{\hat{c}} \sim 0$,
	$\frac{d}{dx}\cos^{-1}(\sqrt{x})\vert_{x=\hat{R}_t/\hat{c}}$ is
	almost a constant, and $\abs{\hat{\theta}_t-\theta_t} = \mathcal{O}
	\qty(\frac{\hat{\epsilon}}{\hat{c}})\leq\delta$. This is usually true
	since $0\leq \tilde{R}_t\leq 1$ and $\hat{c} \gg 1$. Note that
	$\hat{c} =\max\qty{W^x_{\pmb{s}_t}}=\prod^T_{t=1}
	\max\qty{\frac{1}{2\tilde{R}_t},\frac{1}{2(1-\tilde{R}_t)}}$, and
	$\max\qty{\frac{1}{2\tilde{R}_t},\frac{1}{2(1-\tilde{R}_t)}}\geq 1$.
	To make sure
	$\abs{\hat{c}\cos^2(\hat{\theta}_t)-\hat{R}_t} \leq \epsilon$, or
	equivalently $\abs{\hat{\epsilon}}\leq\epsilon$, the optimal $\delta$
	can be chosen is $\mathcal{O}\qty(\frac{\epsilon}{\hat{c}})$. This
	gives $k = \mathcal{O} \qty(\frac{\hat{c}}{\epsilon})$.

	Moreover, for $t^\text{th}$ iteration, the step \eqref{eq:Query}
	requires $t$ queries. So the query complexity for each iteration is
	$\mathcal{O}\qty(\frac{\hat{c}}{\epsilon}t)$.

	Nevertheless, in order to obtain the value of $\tilde{R}_t$, each
	quantum iteration $t$ is followed with a measurement. The information
	of $W^x_{s_t}$ saved in superposition would be disrupted and thus it
	have to be evaluated from every beginning every time. Therefore the
	overall complexity is
	$\mathcal{O}
	(
	\frac
	{\hat{c}}
	{\epsilon}\sum_{t=1}^Tt
	)
	=
	\mathcal{O}(\frac{\hat{c}}{\epsilon}T^2)$,
	comparing to the classical case, which is
	$\mathcal{O}\left(\frac{\hat{c}^2}{\epsilon^2}T\right)$. As
	discussed in remark at the end of the
	section~\hyperref[sec:ClassA]{Proof of Theorem~\ref{thm:ProbOS}},
	AdaBoost algorithm may not work well if it does not converge within a
	small number of iterations. Therefore, the $T$ here may be considered
	as a small constant.

	Also, for both quantum and classical algorithms, we use
	$N=\mathcal{O}\qty(\frac{\hat{c}^2}{\epsilon^2})$, the query complexity
	of classical algorithm can be rewritten as $\mathcal{O}(NT)$ and the
	quantum query complexity is then $\mathcal{O}(\sqrt{N}T^2)$.

	This quantum algorithm could give the same result of the classical
	algorithm with the same order of precision $\epsilon$ with same success
	probability.

	\end{appendix}

	\end{document}